\documentclass[aps, prl, twocolumn]{revtex4}

\usepackage{rotating}
\usepackage{floatpag}
\usepackage{xifthen}
\usepackage[normalem]{ulem}

\usepackage{bm}

\newcommand{\warningsign}{\raisebox{1pt}{\fontencoding{U}\fontfamily{futs}\selectfont\char 66\relax}\xspace}

\ifthenelse{\isundefined{\showoptional}}{
  \newcommand{\showoptional}{1}
}{}

\ifthenelse{\isundefined{\ismain}}{
  \newcommand{\ismain}{0}
}{}

\usepackage{hyperref}

\usepackage{amsmath,amsthm,amssymb}
\usepackage{wasysym}

\usepackage{xspace,enumerate,color,epsfig} 
\usepackage{graphicx}
\graphicspath{{.}{./figures/}}

\usepackage{tikzfig}
\usepackage{stmaryrd}
\usepackage{docmute}
\usepackage{keycommand}

\newcommand{\past}{\textbf{past}}

\newcommand{\inp}[1]{#1^{i}}
\newcommand{\outp}[1]{#1^{o}}

\tikzstyle{env}=[copoint,regular polygon rotate=0,minimum width=0.2cm, fill=black]

\tikzstyle{probs}=[shape=semicircle,fill=white,draw=black,shape border rotate=180,minimum width=1.2cm]

\tikzstyle{every picture}=[baseline=-0.25em,scale=0.5]
\tikzstyle{dotpic}=[] 
\tikzstyle{diredges}=[every to/.style={diredge}]
\tikzstyle{math matrix}=[matrix of math nodes,left delimiter=(,right delimiter=),inner sep=2pt,column sep=1em,row sep=0.5em,nodes={inner sep=0pt},text height=1.5ex, text depth=0.25ex]

\tikzstyle{inline text}=[text height=1.5ex, text depth=0.25ex,yshift=0.5mm]
\tikzstyle{label}=[font=\footnotesize,text height=1.5ex, text depth=0.25ex,yshift=0.5mm]
\tikzstyle{left label}=[label,anchor=east,xshift=1.5mm]
\tikzstyle{right label}=[label,anchor=west,xshift=-1.5mm]

\tikzstyle{braceedge}=[decorate,decoration={brace,amplitude=2mm,raise=-1mm}]
\tikzstyle{small braceedge}=[decorate,decoration={brace,amplitude=1mm,raise=-1mm}]

\tikzstyle{doubled}=[line width=1.6pt] 
\tikzstyle{boldedge}=[doubled,shorten <=-0.17mm,shorten >=-0.17mm]
\tikzstyle{boldedgegray}=[doubled,gray,shorten <=-0.17mm,shorten >=-0.17mm]
\tikzstyle{singleedgegray}=[gray]

\tikzstyle{semidoubled}=[line width=1.4pt] 
\tikzstyle{semiboldedgegray}=[semidoubled,gray,shorten <=-0.17mm,shorten >=-0.17mm]

\tikzstyle{boxedge}=[semiboldedgegray]

\tikzstyle{boldedgedashed}=[very thick,dashed,shorten <=-0.17mm,shorten >=-0.17mm]
\tikzstyle{vboldedgedashed}=[doubled,dashed,shorten <=-0.17mm,shorten >=-0.17mm]
\tikzstyle{left hook arrow}=[left hook-latex]
\tikzstyle{right hook arrow}=[right hook-latex]
\tikzstyle{sembracket}=[line width=0.5pt,shorten <=-0.07mm,shorten >=-0.07mm]

\tikzstyle{causal edge}=[->,thick,gray]
\tikzstyle{causal nondir}=[thick,gray]
\tikzstyle{timeline}=[thick,gray, dashed]

\tikzstyle{cedge}=[<->,thick,gray!70!white]

\tikzstyle{empty diagram}=[draw=gray!40!white,dashed,shape=rectangle,minimum width=1cm,minimum height=1cm]
\tikzstyle{empty diagram small}=[draw=gray!50!white,dashed,shape=rectangle,minimum width=0.6cm,minimum height=0.5cm]

\tikzstyle{dot}=[inner sep=0mm,minimum width=2mm,minimum height=2mm,draw,shape=circle]  
\tikzstyle{Wsquare}=[white dot, shape=regular polygon, rounded corners=0.8 mm, minimum size=3.3 mm, regular polygon sides=3, outer sep=-0.2mm]
\tikzstyle{Wsquareadj}=[white dot, shape=regular polygon, rounded corners=0.8 mm, minimum size=3.3 mm, regular polygon sides=3, outer sep=-0.2mm, regular polygon rotate=180]
\tikzstyle{ddot}=[inner sep=0mm, doubled, minimum width=2.5mm,minimum height=2.5mm,draw,shape=circle]

\tikzstyle{black dot}=[dot,fill=black]
\tikzstyle{white dot}=[dot,fill=white,,text depth=-0.2mm]
\tikzstyle{white Wsquare}=[Wsquare,fill=gray,,text depth=-0.2mm]
\tikzstyle{white Wsquareadj}=[Wsquareadj,fill=white,,text depth=-0.2mm]
\tikzstyle{green dot}=[white dot] 
\tikzstyle{gray dot}=[dot,fill=gray!40!white,,text depth=-0.2mm]
\tikzstyle{red dot}=[gray dot] 

\tikzstyle{black ddot}=[ddot,fill=black]
\tikzstyle{white ddot}=[ddot,fill=white]
\tikzstyle{gray ddot}=[ddot,fill=gray!40!white]

\tikzstyle{gray edge}=[gray!60!white]

\tikzstyle{small dot}=[inner sep=0.5mm,minimum width=0pt,minimum height=0pt,draw,shape=circle]

\tikzstyle{small black dot}=[small dot,fill=black]
\tikzstyle{small white dot}=[small dot,fill=white]
\tikzstyle{small gray dot}=[small dot,fill=gray!40!white]

\tikzstyle{causal dot}=[inner sep=0.4mm,minimum width=0pt,minimum height=0pt,draw=white,shape=circle,fill=gray!40!white]

\tikzstyle{phase dimensions}=[minimum size=5mm,font=\footnotesize,rectangle,rounded corners=2.5mm,inner sep=0.2mm,outer sep=-2mm]
\tikzstyle{dphase dimensions}=[minimum size=5mm,font=\footnotesize,rectangle,rounded corners=2.5mm,inner sep=0.2mm,outer sep=-2mm]

\tikzstyle{white phase dot}=[dot,fill=white,phase dimensions]
\tikzstyle{white phase ddot}=[ddot,fill=white,dphase dimensions]

\tikzstyle{white rect ddot}=[draw=black,fill=white,doubled,minimum size=5mm,font=\footnotesize,rectangle,rounded corners=2.5mm,inner sep=0.2mm]
\tikzstyle{gray rect ddot}=[draw=black,fill=gray!40!white,doubled,minimum size=6mm,font=\footnotesize,rectangle,rounded corners=3mm]

\tikzstyle{gray phase dot}=[dot,fill=gray!40!white,phase dimensions]
\tikzstyle{gray phase ddot}=[ddot,fill=gray!40!white,dphase dimensions]
\tikzstyle{grey phase dot}=[gray phase dot]
\tikzstyle{grey phase ddot}=[gray phase ddot]

\tikzstyle{small phase dimensions}=[minimum size=4mm,font=\tiny,rectangle,rounded corners=2mm,inner sep=0.2mm,outer sep=-2mm]
\tikzstyle{small dphase dimensions}=[minimum size=4mm,font=\tiny,rectangle,rounded corners=2mm,inner sep=0.2mm,outer sep=-2mm]

\tikzstyle{small gray phase dot}=[dot,fill=gray!40!white,small phase dimensions]
\tikzstyle{small gray phase ddot}=[ddot,fill=gray!40!white,small dphase dimensions]

\tikzstyle{small map}=[draw,shape=rectangle,minimum height=4mm,minimum width=4mm,fill=white]

\tikzstyle{cnot}=[fill=white,shape=circle,inner sep=-1.4pt]

\tikzstyle{asym hadamard}=[fill=white,draw,shape=NEbox,inner sep=0.6mm,font=\footnotesize,minimum height=4mm]
\tikzstyle{asym hadamard conj}=[fill=white,draw,shape=NWbox,inner sep=0.6mm,font=\footnotesize,minimum height=4mm]
\tikzstyle{asym hadamard dag}=[fill=white,draw,shape=SEbox,inner sep=0.6mm,font=\footnotesize,minimum height=4mm]

\tikzstyle{hadamard}=[fill=white,draw,inner sep=0.6mm,font=\footnotesize,minimum height=4mm,minimum width=4mm]
\tikzstyle{small hadamard}=[fill=white,draw,inner sep=0.6mm,minimum height=1.5mm,minimum width=1.5mm]
\tikzstyle{small hadamard rotate}=[small hadamard,rotate=45]
\tikzstyle{dhadamard}=[hadamard,doubled]
\tikzstyle{small dhadamard}=[small hadamard,doubled]
\tikzstyle{small dhadamard rotate}=[small hadamard rotate,doubled]
\tikzstyle{antipode}=[white dot,inner sep=0.3mm,font=\footnotesize]

\tikzstyle{scalar}=[diamond,draw,inner sep=0.5pt,font=\small]
\tikzstyle{dscalar}=[diamond,doubled, draw,inner sep=0.5pt,font=\small]

\tikzstyle{small box}=[rectangle,inline text,fill=white,draw,minimum height=5mm,yshift=-0.5mm,minimum width=5mm,font=\small]
\tikzstyle{small gray box}=[small box,fill=gray!30]
\tikzstyle{medium box}=[rectangle,inline text,fill=white,draw,minimum height=5mm,yshift=-0.5mm,minimum width=10mm,font=\small]
\tikzstyle{square box}=[small box] 
\tikzstyle{medium gray box}=[small box,fill=gray!30]
\tikzstyle{semilarge box}=[rectangle,inline text,fill=white,draw,minimum height=5mm,yshift=-0.5mm,minimum width=12.5mm,font=\small]
\tikzstyle{large box}=[rectangle,inline text,fill=white,draw,minimum height=5mm,yshift=-0.5mm,minimum width=15mm,font=\small]
\tikzstyle{large gray box}=[small box,fill=gray!30]

\tikzstyle{Bayes box}=[rectangle,fill=black,draw, minimum height=3mm, minimum width=3mm]

\tikzstyle{gray square point}=[small box,fill=gray!50]

\tikzstyle{dphase box white}=[dhadamard]
\tikzstyle{dphase box gray}=[dhadamard,fill=gray!50!white]
\tikzstyle{phase box white}=[hadamard]
\tikzstyle{phase box gray}=[hadamard,fill=gray!50!white]

\tikzstyle{point}=[regular polygon,regular polygon sides=3,draw,scale=0.75,inner sep=-0.5pt,minimum width=9mm,fill=white,regular polygon rotate=180]
\tikzstyle{point nosep}=[regular polygon,regular polygon sides=3,draw,scale=0.75,inner sep=-2pt,minimum width=9mm,fill=white,regular polygon rotate=180]
\tikzstyle{copoint}=[regular polygon,regular polygon sides=3,draw,scale=0.75,inner sep=-0.5pt,minimum width=9mm,fill=white]
\tikzstyle{dpoint}=[point,doubled]
\tikzstyle{dcopoint}=[copoint,doubled]

\tikzstyle{pointgrow}=[shape=cornerpoint,kpoint common,scale=0.75,inner sep=3pt]
\tikzstyle{pointgrow dag}=[shape=cornercopoint,kpoint common,scale=0.75,inner sep=3pt]

\tikzstyle{wide copoint}=[fill=white,draw,shape=isosceles triangle,shape border rotate=90,isosceles triangle stretches=true,inner sep=0pt,minimum width=1.5cm,minimum height=6.12mm]
\tikzstyle{wide point}=[fill=white,draw,shape=isosceles triangle,shape border rotate=-90,isosceles triangle stretches=true,inner sep=0pt,minimum width=1.5cm,minimum height=6.12mm,yshift=-0.0mm]
\tikzstyle{wide point plus}=[fill=white,draw,shape=isosceles triangle,shape border rotate=-90,isosceles triangle stretches=true,inner sep=0pt,minimum width=1.74cm,minimum height=7mm,yshift=-0.0mm]

\tikzstyle{wide dpoint}=[fill=white,doubled,draw,shape=isosceles triangle,shape border rotate=-90,isosceles triangle stretches=true,inner sep=0pt,minimum width=1.5cm,minimum height=6.12mm,yshift=-0.0mm]

\tikzstyle{tinypoint}=[regular polygon,regular polygon sides=3,draw,scale=0.55,inner sep=-0.15pt,minimum width=6mm,fill=white,regular polygon rotate=180] 

\tikzstyle{white point}=[point]
\tikzstyle{white dpoint}=[dpoint]
\tikzstyle{green point}=[white point] 
\tikzstyle{white copoint}=[copoint]
\tikzstyle{gray point}=[point,fill=gray!40!white]
\tikzstyle{gray dpoint}=[gray point,doubled]
\tikzstyle{red point}=[gray point] 
\tikzstyle{gray copoint}=[copoint,fill=gray!40!white]
\tikzstyle{gray dcopoint}=[gray copoint,doubled]

\tikzstyle{white point guide}=[regular polygon,regular polygon sides=3,font=\scriptsize,draw,scale=0.65,inner sep=-0.5pt,minimum width=9mm,fill=white,regular polygon rotate=180]

\tikzstyle{black point}=[point,fill=black,font=\color{white}]
\tikzstyle{black copoint}=[copoint,fill=black,font=\color{white}]

\tikzstyle{tiny gray point}=[tinypoint,fill=gray!40!white]

\tikzstyle{diredge}=[->]
\tikzstyle{ddiredge}=[<->]
\tikzstyle{rdiredge}=[<-]
\tikzstyle{thickdiredge}=[->, very thick]
\tikzstyle{pointer edge}=[->,very thick,gray]
\tikzstyle{pointer edge part}=[very thick,gray]
\tikzstyle{dashed edge}=[dashed]
\tikzstyle{thick dashed edge}=[very thick,dashed]
\tikzstyle{thick gray dashed edge}=[thick dashed edge,gray!40]
\tikzstyle{thick map edge}=[very thick,|->]

\makeatletter
\newcommand{\boxshape}[3]{%
\pgfdeclareshape{#1}{
\inheritsavedanchors[from=rectangle] 
\inheritanchorborder[from=rectangle]
\inheritanchor[from=rectangle]{center}
\inheritanchor[from=rectangle]{north}
\inheritanchor[from=rectangle]{south}
\inheritanchor[from=rectangle]{west}
\inheritanchor[from=rectangle]{east}
\backgroundpath{
\southwest \pgf@xa=\pgf@x \pgf@ya=\pgf@y
\northeast \pgf@xb=\pgf@x \pgf@yb=\pgf@y

\@tempdima=#2
\@tempdimb=#3

\pgfpathmoveto{\pgfpoint{\pgf@xa - 5pt + \@tempdima}{\pgf@ya}}
\pgfpathlineto{\pgfpoint{\pgf@xa - 5pt - \@tempdima}{\pgf@yb}}
\pgfpathlineto{\pgfpoint{\pgf@xb + 5pt + \@tempdimb}{\pgf@yb}}
\pgfpathlineto{\pgfpoint{\pgf@xb + 5pt - \@tempdimb}{\pgf@ya}}
\pgfpathlineto{\pgfpoint{\pgf@xa - 5pt + \@tempdima}{\pgf@ya}}
\pgfpathclose
}
}}

\boxshape{NEbox}{0pt}{5pt}
\boxshape{SEbox}{0pt}{-5pt}
\boxshape{NWbox}{5pt}{0pt}
\boxshape{SWbox}{-5pt}{0pt}
\boxshape{EBox}{-3pt}{3pt}
\boxshape{WBox}{3pt}{-3pt}
\makeatother

\tikzstyle{cloud}=[shape=cloud,draw,minimum width=1.5cm,minimum height=1.5cm]

\tikzstyle{map}=[draw,shape=NEbox,inner sep=2pt,minimum height=6mm,fill=white]
\tikzstyle{dashedmap}=[draw,dashed,shape=NEbox,inner sep=2pt,minimum height=6mm,fill=white]
\tikzstyle{mapdag}=[draw,shape=SEbox,inner sep=2pt,minimum height=6mm,fill=white]
\tikzstyle{mapadj}=[draw,shape=SEbox,inner sep=2pt,minimum height=6mm,fill=white]
\tikzstyle{maptrans}=[draw,shape=SWbox,inner sep=2pt,minimum height=6mm,fill=white]
\tikzstyle{mapconj}=[draw,shape=NWbox,inner sep=2pt,minimum height=6mm,fill=white]

\tikzstyle{medium map}=[draw,shape=NEbox,inner sep=2pt,minimum height=6mm,fill=white,minimum width=7mm]
\tikzstyle{medium map dag}=[draw,shape=SEbox,inner sep=2pt,minimum height=6mm,fill=white,minimum width=7mm]
\tikzstyle{medium map adj}=[draw,shape=SEbox,inner sep=2pt,minimum height=6mm,fill=white,minimum width=7mm]
\tikzstyle{medium map trans}=[draw,shape=SWbox,inner sep=2pt,minimum height=6mm,fill=white,minimum width=7mm]
\tikzstyle{medium map conj}=[draw,shape=NWbox,inner sep=2pt,minimum height=6mm,fill=white,minimum width=7mm]
\tikzstyle{semilarge map}=[draw,shape=NEbox,inner sep=2pt,minimum height=6mm,fill=white,minimum width=9.5mm]
\tikzstyle{semilarge map trans}=[draw,shape=SWbox,inner sep=2pt,minimum height=6mm,fill=white,minimum width=9.5mm]
\tikzstyle{semilarge map adj}=[draw,shape=SEbox,inner sep=2pt,minimum height=6mm,fill=white,minimum width=9.5mm]
\tikzstyle{semilarge map dag}=[draw,shape=SEbox,inner sep=2pt,minimum height=6mm,fill=white,minimum width=9.5mm]
\tikzstyle{semilarge map conj}=[draw,shape=NWbox,inner sep=2pt,minimum height=6mm,fill=white,minimum width=9.5mm]
\tikzstyle{large map}=[draw,shape=NEbox,inner sep=2pt,minimum height=6mm,fill=white,minimum width=12mm]
\tikzstyle{large map conj}=[draw,shape=NWbox,inner sep=2pt,minimum height=6mm,fill=white,minimum width=12mm]
\tikzstyle{very large map}=[draw,shape=NEbox,inner sep=2pt,minimum height=6mm,fill=white,minimum width=17mm]

\tikzstyle{medium dmap}=[draw,doubled,shape=NEbox,inner sep=2pt,minimum height=6mm,fill=white,minimum width=7mm]
\tikzstyle{medium dmap dag}=[draw,doubled,shape=SEbox,inner sep=2pt,minimum height=6mm,fill=white,minimum width=7mm]
\tikzstyle{medium dmap adj}=[draw,doubled,shape=SEbox,inner sep=2pt,minimum height=6mm,fill=white,minimum width=7mm]
\tikzstyle{medium dmap trans}=[draw,doubled,shape=SWbox,inner sep=2pt,minimum height=6mm,fill=white,minimum width=7mm]
\tikzstyle{medium dmap conj}=[draw,doubled,shape=NWbox,inner sep=2pt,minimum height=6mm,fill=white,minimum width=7mm]
\tikzstyle{semilarge dmap}=[draw,doubled,shape=NEbox,inner sep=2pt,minimum height=6mm,fill=white,minimum width=9.5mm]
\tikzstyle{semilarge dmap trans}=[draw,doubled,shape=SWbox,inner sep=2pt,minimum height=6mm,fill=white,minimum width=9.5mm]
\tikzstyle{semilarge dmap adj}=[draw,doubled,shape=SEbox,inner sep=2pt,minimum height=6mm,fill=white,minimum width=9.5mm]
\tikzstyle{semilarge dmap dag}=[draw,doubled,shape=SEbox,inner sep=2pt,minimum height=6mm,fill=white,minimum width=9.5mm]
\tikzstyle{semilarge dmap conj}=[draw,doubled,shape=NWbox,inner sep=2pt,minimum height=6mm,fill=white,minimum width=9.5mm]
\tikzstyle{large dmap}=[draw,doubled,shape=NEbox,inner sep=2pt,minimum height=6mm,fill=white,minimum width=12mm]
\tikzstyle{large dmap conj}=[draw,doubled,shape=NWbox,inner sep=2pt,minimum height=6mm,fill=white,minimum width=12mm]
\tikzstyle{large dmap trans}=[draw,doubled,shape=SWbox,inner sep=2pt,minimum height=6mm,fill=white,minimum width=12mm]
\tikzstyle{large dmap adj}=[draw,doubled,shape=SEbox,inner sep=2pt,minimum height=6mm,fill=white,minimum width=12mm]
\tikzstyle{large dmap dag}=[draw,doubled,shape=SEbox,inner sep=2pt,minimum height=6mm,fill=white,minimum width=12mm]
\tikzstyle{very large dmap}=[draw,doubled,shape=NEbox,inner sep=2pt,minimum height=6mm,fill=white,minimum width=19.5mm]

\tikzstyle{muxbox}=[draw,shape=rectangle,minimum height=3mm,minimum width=3mm,fill=white]
\tikzstyle{dmuxbox}=[muxbox,doubled]

\tikzstyle{box}=[draw,shape=rectangle,inner sep=2pt,minimum height=6mm,minimum width=6mm,fill=white]
\tikzstyle{dbox}=[draw,doubled,shape=rectangle,inner sep=2pt,minimum height=6mm,minimum width=6mm,fill=white]
\tikzstyle{dmap}=[draw,doubled,shape=NEbox,inner sep=2pt,minimum height=6mm,fill=white]
\tikzstyle{dmapdag}=[draw,doubled,shape=SEbox,inner sep=2pt,minimum height=6mm,fill=white]
\tikzstyle{dmapadj}=[draw,doubled,shape=SEbox,inner sep=2pt,minimum height=6mm,fill=white]
\tikzstyle{dmaptrans}=[draw,doubled,shape=SWbox,inner sep=2pt,minimum height=6mm,fill=white]
\tikzstyle{dmapconj}=[draw,doubled,shape=NWbox,inner sep=2pt,minimum height=6mm,fill=white]

\tikzstyle{ddmap}=[draw,doubled,dashed,shape=NEbox,inner sep=2pt,minimum height=6mm,fill=white]
\tikzstyle{ddmapdag}=[draw,doubled,dashed,shape=SEbox,inner sep=2pt,minimum height=6mm,fill=white]
\tikzstyle{ddmapadj}=[draw,doubled,dashed,shape=SEbox,inner sep=2pt,minimum height=6mm,fill=white]
\tikzstyle{ddmaptrans}=[draw,doubled,dashed,shape=SWbox,inner sep=2pt,minimum height=6mm,fill=white]
\tikzstyle{ddmapconj}=[draw,doubled,dashed,shape=NWbox,inner sep=2pt,minimum height=6mm,fill=white]

\boxshape{sNEbox}{0pt}{3pt}
\boxshape{sSEbox}{0pt}{-3pt}
\boxshape{sNWbox}{3pt}{0pt}
\boxshape{sSWbox}{-3pt}{0pt}
\tikzstyle{smap}=[draw,shape=sNEbox,fill=white]
\tikzstyle{smapdag}=[draw,shape=sSEbox,fill=white]
\tikzstyle{smapadj}=[draw,shape=sSEbox,fill=white]
\tikzstyle{smaptrans}=[draw,shape=sSWbox,fill=white]
\tikzstyle{smapconj}=[draw,shape=sNWbox,fill=white]

\tikzstyle{dsmap}=[draw,dashed,shape=sNEbox,fill=white]
\tikzstyle{dsmapdag}=[draw,dashed,shape=sSEbox,fill=white]
\tikzstyle{dsmaptrans}=[draw,dashed,shape=sSWbox,fill=white]
\tikzstyle{dsmapconj}=[draw,dashed,shape=sNWbox,fill=white]

\boxshape{mNEbox}{0pt}{10pt}
\boxshape{mSEbox}{0pt}{-10pt}
\boxshape{mNWbox}{10pt}{0pt}
\boxshape{mSWbox}{-10pt}{0pt}
\tikzstyle{mmap}=[draw,shape=mNEbox]
\tikzstyle{mmapdag}=[draw,shape=mSEbox]
\tikzstyle{mmaptrans}=[draw,shape=mSWbox]
\tikzstyle{mmapconj}=[draw,shape=mNWbox]

\tikzstyle{mmapgray}=[draw,fill=gray!40!white,shape=mNEbox]
\tikzstyle{smapgray}=[draw,fill=gray!40!white,shape=sNEbox]

\makeatletter

\pgfdeclareshape{cornerpoint}{
\inheritsavedanchors[from=rectangle] 
\inheritanchorborder[from=rectangle]
\inheritanchor[from=rectangle]{center}
\inheritanchor[from=rectangle]{north}
\inheritanchor[from=rectangle]{south}
\inheritanchor[from=rectangle]{west}
\inheritanchor[from=rectangle]{east}
\backgroundpath{
\southwest \pgf@xa=\pgf@x \pgf@ya=\pgf@y
\northeast \pgf@xb=\pgf@x \pgf@yb=\pgf@y

\pgfmathsetmacro{\pgf@shorten@left}{\pgfkeysvalueof{/tikz/shorten left}}
\pgfmathsetmacro{\pgf@shorten@right}{\pgfkeysvalueof{/tikz/shorten right}}

\pgfpathmoveto{\pgfpoint{0.5 * (\pgf@xa + \pgf@xb)}{\pgf@ya - 5pt}}
\pgfpathlineto{\pgfpoint{\pgf@xa - 8pt + \pgf@shorten@left}{\pgf@yb - 1.5 * \pgf@shorten@left}}
\pgfpathlineto{\pgfpoint{\pgf@xa - 8pt + \pgf@shorten@left}{\pgf@yb}}
\pgfpathlineto{\pgfpoint{\pgf@xb + 8pt - \pgf@shorten@right}{\pgf@yb}}
\pgfpathlineto{\pgfpoint{\pgf@xb + 8pt - \pgf@shorten@right}{\pgf@yb - 1.5 * \pgf@shorten@right}}
\pgfpathclose
}
}

\pgfdeclareshape{cornercopoint}{
\inheritsavedanchors[from=rectangle] 
\inheritanchorborder[from=rectangle]
\inheritanchor[from=rectangle]{center}
\inheritanchor[from=rectangle]{north}
\inheritanchor[from=rectangle]{south}
\inheritanchor[from=rectangle]{west}
\inheritanchor[from=rectangle]{east}
\backgroundpath{
\southwest \pgf@xa=\pgf@x \pgf@ya=\pgf@y
\northeast \pgf@xb=\pgf@x \pgf@yb=\pgf@y

\pgfmathsetmacro{\pgf@shorten@left}{\pgfkeysvalueof{/tikz/shorten left}}
\pgfmathsetmacro{\pgf@shorten@right}{\pgfkeysvalueof{/tikz/shorten right}}

\pgfpathmoveto{\pgfpoint{0.5 * (\pgf@xa + \pgf@xb)}{\pgf@yb + 5pt}}
\pgfpathlineto{\pgfpoint{\pgf@xa - 8pt + \pgf@shorten@left}{\pgf@ya + 1.5 * \pgf@shorten@left}}
\pgfpathlineto{\pgfpoint{\pgf@xa - 8pt + \pgf@shorten@left}{\pgf@ya}}
\pgfpathlineto{\pgfpoint{\pgf@xb + 8pt - \pgf@shorten@right}{\pgf@ya}}
\pgfpathlineto{\pgfpoint{\pgf@xb + 8pt - \pgf@shorten@right}{\pgf@ya + 1.5 * \pgf@shorten@right}}
\pgfpathclose
}
}

\makeatother

\pgfkeyssetvalue{/tikz/shorten left}{0pt}
\pgfkeyssetvalue{/tikz/shorten right}{0pt}

\tikzstyle{kpoint common}=[draw,fill=white,inner sep=1pt,minimum height=4mm]
\tikzstyle{kpoint sc}=[shape=cornerpoint,kpoint common]
\tikzstyle{kpoint adjoint sc}=[shape=cornercopoint,kpoint common]
\tikzstyle{kpoint}=[shape=cornerpoint,shorten left=5pt,kpoint common]
\tikzstyle{kpoint adjoint}=[shape=cornercopoint,shorten left=5pt,kpoint common]
\tikzstyle{kpoint conjugate}=[shape=cornerpoint,shorten right=5pt,kpoint common]
\tikzstyle{kpoint transpose}=[shape=cornercopoint,shorten right=5pt,kpoint common]
\tikzstyle{kpoint symm}=[shape=cornerpoint,shorten left=5pt,shorten right=5pt,kpoint common]

\tikzstyle{wide kpoint sc}=[shape=cornerpoint,kpoint common, minimum width=1 cm]
\tikzstyle{wide kpointdag sc}=[shape=cornercopoint,kpoint common, minimum width=1 cm]

\tikzstyle{black kpoint}=[shape=cornerpoint,shorten left=5pt,kpoint common,fill=black,font=\color{white}]

\tikzstyle{black kpoint sm}=[shape=cornerpoint,shorten left=5pt,kpoint common,fill=black,font=\color{white},scale=0.75]

\tikzstyle{black kpoint adjoint}=[shape=cornercopoint,shorten left=5pt,kpoint common,fill=black,font=\color{white}]
\tikzstyle{black kpointadj}=[shape=cornercopoint,shorten left=5pt,kpoint common,fill=black,font=\color{white}]

\tikzstyle{black kpointadj sm}=[shape=cornercopoint,shorten left=5pt,kpoint common,fill=black,font=\color{white},scale=0.75]

\tikzstyle{black dkpoint}=[shape=cornerpoint,shorten left=5pt,kpoint common,fill=black, doubled,font=\color{white}]
\tikzstyle{black dkpoint adjoint}=[shape=cornercopoint,shorten left=5pt,kpoint common,fill=black, doubled,font=\color{white}]
\tikzstyle{black dkpointadj}=[shape=cornercopoint,shorten left=5pt,kpoint common,fill=black, doubled,font=\color{white}]

\tikzstyle{black dkpoint sm}=[shape=cornerpoint,shorten left=5pt,kpoint common,fill=black, doubled,font=\color{white},scale=0.75]
\tikzstyle{black dkpointadj sm}=[shape=cornercopoint,shorten left=5pt,kpoint common,fill=black, doubled,font=\color{white},scale=0.75] 

\tikzstyle{kpointdag}=[kpoint adjoint]
\tikzstyle{kpointadj}=[kpoint adjoint]
\tikzstyle{kpointconj}=[kpoint conjugate]
\tikzstyle{kpointtrans}=[kpoint transpose]

\tikzstyle{big kpoint}=[kpoint, minimum width=1.2 cm, minimum height=8mm, inner sep=4pt, text depth=3mm]

\tikzstyle{wide kpoint}=[kpoint, minimum width=1 cm, inner sep=2pt]
\tikzstyle{wide kpointdag}=[kpointdag, minimum width=1 cm, inner sep=2pt]
\tikzstyle{wide kpointconj}=[kpointconj, minimum width=1 cm, inner sep=2pt]
\tikzstyle{wide kpointtrans}=[kpointtrans, minimum width=1 cm, inner sep=2pt]

\tikzstyle{wider kpoint}=[kpoint, minimum width=1.25 cm, inner sep=2pt]
\tikzstyle{wider kpointdag}=[kpointdag, minimum width=1.25 cm, inner sep=2pt]
\tikzstyle{wider kpointconj}=[kpointconj, minimum width=1.25 cm, inner sep=2pt]
\tikzstyle{wider kpointtrans}=[kpointtrans, minimum width=1.25 cm, inner sep=2pt]

\tikzstyle{gray kpoint}=[kpoint,fill=gray!50!white]
\tikzstyle{gray kpointdag}=[kpointdag,fill=gray!50!white]
\tikzstyle{gray kpointadj}=[kpointadj,fill=gray!50!white]
\tikzstyle{gray kpointconj}=[kpointconj,fill=gray!50!white]
\tikzstyle{gray kpointtrans}=[kpointtrans,fill=gray!50!white]

\tikzstyle{gray dkpoint}=[kpoint,fill=gray!50!white,doubled]
\tikzstyle{gray dkpointdag}=[kpointdag,fill=gray!50!white,doubled]
\tikzstyle{gray dkpointadj}=[kpointadj,fill=gray!50!white,doubled]
\tikzstyle{gray dkpointconj}=[kpointconj,fill=gray!50!white,doubled]
\tikzstyle{gray dkpointtrans}=[kpointtrans,fill=gray!50!white,doubled]

\tikzstyle{white label}=[draw,fill=white,rectangle,inner sep=0.7 mm]
\tikzstyle{gray label}=[draw,fill=gray!50!white,rectangle,inner sep=0.7 mm]
\tikzstyle{black label}=[draw,fill=black,rectangle,inner sep=0.7 mm]

\tikzstyle{dkpoint}=[kpoint,doubled]
\tikzstyle{wide dkpoint}=[wide kpoint,doubled]
\tikzstyle{dkpointdag}=[kpoint adjoint,doubled]
\tikzstyle{wide dkpointdag}=[wide kpointdag,doubled]
\tikzstyle{dkcopoint}=[kpoint adjoint,doubled]
\tikzstyle{dkpointadj}=[kpoint adjoint,doubled]
\tikzstyle{dkpointconj}=[kpoint conjugate,doubled]
\tikzstyle{dkpointtrans}=[kpoint transpose,doubled]

\tikzstyle{kscalar}=[kpoint common, shape=EBox, inner xsep=-1pt, inner ysep=3pt,font=\small]
\tikzstyle{kscalarconj}=[kpoint common, shape=WBox, inner xsep=-1pt, inner ysep=3pt,font=\small]

\tikzstyle{spekpoint}=[kpoint sc,minimum height=5mm,inner sep=3pt]
\tikzstyle{spekcopoint}=[kpoint adjoint sc,minimum height=5mm,inner sep=3pt]

\tikzstyle{dspekpoint}=[spekpoint,doubled]
\tikzstyle{dspekcopoint}=[spekcopoint,doubled]

 \tikzstyle{upground}=[circuit ee IEC,thick,ground,rotate=90,scale=1.4]
 \tikzstyle{upgroundnormal}=[circuit ee IEC,thick,ground,rotate=90,scale=2]
 \tikzstyle{downground}=[circuit ee IEC,thick,ground,rotate=-90,scale=1.4]
 \tikzstyle{bigground}=[regular polygon,regular polygon sides=3,draw=gray,scale=0.50,inner sep=-0.5pt,minimum width=10mm,fill=gray]

\tikzstyle{arrs}=[-latex,font=\small,auto]
\tikzstyle{arrow plain}=[arrs]
\tikzstyle{arrow dashed}=[dashed,arrs]
\tikzstyle{arrow bold}=[very thick,arrs]
\tikzstyle{arrow hide}=[draw=white!0,-]
\tikzstyle{arrow reverse}=[latex-]
\tikzstyle{cdnode}=[]

\let\olddagger\dagger
\renewcommand{\dagger}{\ensuremath{\olddagger}\xspace}

\usepackage{makeidx}
\makeindex

\theoremstyle{definition}

\newtheorem*{theorem*}{Theorem}

\newtheorem{example*}[theorem]{Example*}
\newtheorem{examples*}[theorem]{Examples*}

\newtheorem{remark*}[theorem]{Remark*}

\newtheorem{exer*}[theorem]{Exercise*}

\newkeycommand{\moral}[width=11cm][1]{\begin{center}
\fbox{\ \parbox{\commandkey{width}}{\centering #1\vphantom{Xy}}\ } 
\end{center}}

\newkeycommand{\morallong}[width=11cm][1]{\par\medskip\noindent
\centerline{\fbox{\ \parbox{\commandkey{width}}{\centering #1\vphantom{Xy}}\ }} 
\par\medskip\noindent}

\usepackage{color}
\def\bR{\begin{color}{red}} 
\def\bB{\begin{color}{blue}}
\def\bM{\begin{color}{magenta}}
\def\bC{\begin{color}{cyan}}
\def\bW{\begin{color}{white}}
\def\bBl{\begin{color}{black}} 
\def\bG{\begin{color}{green}}
\def\bY{\begin{color}{yellow}}
\def\e{\end{color}\xspace}
\newcommand{\bit}{\begin{itemize}}
\newcommand{\eit}{\end{itemize}\par\noindent}
\newcommand{\ben}{\begin{enumerate}}
\newcommand{\een}{\end{enumerate}\par\noindent}
\newcommand{\beq}{\begin{equation}}
\newcommand{\eeq}{\end{equation}\par\noindent}
\newcommand{\beqa}{\begin{eqnarray*}}
\newcommand{\eeqa}{\end{eqnarray*}\par\noindent}
\newcommand{\beqn}{\begin{eqnarray}}
\newcommand{\eeqn}{\end{eqnarray}\par\noindent}

\begin{document}   
 
\title{Equivalence of relativistic causal structure and process terminality}  
\author{Aleks Kissinger}
\affiliation{Radboud University, iCIS, Nijmegen.}   
\author{Matty Hoban}
\affiliation{University of Edinburgh, School of Informatics.}   
\author{Bob Coecke}
\affiliation{University of Oxford, Department of Computer Science.}      

\begin{abstract}
In general relativity, `causal structure'  refers to the partial order on space-time points (or regions) that encodes time-like relationships.  Recently,  quantum information and quantum foundations saw the emergence of a   `causality principle'.  In the form used in this paper, which we call `process terminality', it states that when the output of a process is discarded, then the process itself may as well be discarded.  While causal structure and process terminality at first seem to be entirely different notions, they become equivalent when making explicit what the partial order actually encodes, that is, that an event in the past can influence one in its future, but not vice-versa. The framework in which we establish this equivalence is that of process theories.
We show how several previous results are instances of this result and comment on how this framework could provide an ideal, minimal canvas for crafting theories of quantum gravity.
\end{abstract} 

\maketitle

\em Causal structure \em is the partial order on spacetime points obtained by setting $x\leq y$ whenever $x$ causally precedes $y$. Starting with the first singularity theorems \cite{penrose1965gravitational, hawking1973large} causal structure has played a prominent role in General Relativity (GR).  In fact, building further on \cite{kronheimer1967structure,penroseorder, Malament}, in \cite{MartinPanagaden} it was shown that for globally hyperbolic spacetimes, the entire spacetime can be reconstructed from the partial order on spacetime points \footnote{In the earlier papers \cite{penroseorder,Malament}, in order to reconstruct the entire spacetime, smoothness of curves needed to be assumed.  Using techniques from domain theory, this assumption was dropped in \cite{MartinPanagaden}.}.  Hence, 
for an important class of spacetimes,  causal structure `is' GR.      

More recently, both within the context of quantum information and quantum foundations, a new principle, called simply \em causality\em, has been coined. Its initial aim was to capture `no-signalling from the future'  within the so-called \em operational probabilistic theories \em framework of Chiribella, D'Ariano and Perinotti~\cite{Chiri1, Chiri2}. A further generalisation of this axiom which makes sense for arbitrary \em process theories \em \cite{JTF, CKpaperI, CKBook} is the following  \cite{CRCaucat, Cnonsig, coecke2016terminality}:  
\begin{equation}\label{eq:term-words}
\begin{array}{c}
\mbox{\em When the output of a process is discarded,}\\ 
\mbox{\em then the process itself may also be discarded.\em} 
\end{array}
\end{equation}
Note that this principle also implies that when doing science we can ignore processes of which the output never reaches us before ending an experiment.  Clearly, such an assumption is absolutely crucial to being able to do science at all, and hence, shouldn't come as a surprise.

To avoid confusion, we refer to  principle (\ref{eq:term-words}) as \em process terminality\em. 
The generality of this definition is important in that, unlike for its instantiation in \cite{Chiri1, Chiri2}, it doesn't require \textit{a priori} the distinction between quantum systems and classical data or any reference to measurements or probabilities, which is also
the case in GR.

Using the diagrammatic language of process theories which we explain below, process terminality is expressed by the following equation:
\begin{equation}\label{eq:term}
  \begin{tikzpicture}
	\begin{pgfonlayer}{nodelayer}
		\node [style=box] (0) at (-1.5, 0) {$f$};
		\node [style=upground] (1) at (-1.5, 1.25) {};
		\node [style=none] (2) at (-1.5, -1.25) {};
		\node [style=none] (3) at (0.25, 0) {$=$};
		\node [style=upground] (4) at (1.5, 0.75) {};
		\node [style=none] (5) at (1.5, -0.75) {};
	\end{pgfonlayer}
	\begin{pgfonlayer}{edgelayer}
		\draw (2.center) to (1);
		\draw (5.center) to (4);
	\end{pgfonlayer}
\end{tikzpicture}
\end{equation}
Here $f$ is a process and the `ground'-symbol represents a special process with no outputs called \textit{discarding}.  Note in particular that there is indeed no reference to measurements, only to the general notion of a process, which could either be classical, quantum, a hybrid of both \cite{CQMII, CKBook}, or even something that has nothing to do with quantum or classical theory.

It has already been shown that from (\ref{eq:term}) one can in fact derive `no-signalling' for two parties  \cite{coecke2016terminality}. This clearly hints at the fact that there is a close connection between (\ref{eq:term}) and the theory of relativity.  

Still, a partially ordered set and axiom (\ref{eq:term}) look very different. However, by explicitly equipping causal structure with its intended interpretation, namely that it encodes forward causation, we can establish  equivalence of causal structure and process terminality within the language of process theories.

It has already been noted in the literature that no-signalling is a strictly weaker notion than relativistic causality \cite{horodecki2016relativistic}. By taking into account the causal structure of measurement events, one can put restrictions on correlations observed in an experiment.
Our work demonstrates how no-signalling  can be generalised as (\ref{eq:term}).

Our work also forms part of an ongoing effort in the study of an operational form of causality in theories beyond quantum theory. On a related note, there have been efforts to generalise the framework of Bayesian networks developed by Pearl and others \cite{PearlBook} to quantum theory and beyond. Henson, Lal, and Pusey demonstrated that certain operational conditional independences (like the no-signalling constraint) hold for classical data in an operational probabilistic theory that satisfies causality; this generalises the notion of d-separation as defined within classical Bayesian networks \cite{HLP}. 

It should be noted that the constraints in the work of Henson, Lal, and Pusey are made upon classical data embedded in a theory, and not objects within a general theory itself. There has been recent, exciting progress on developing a notion of quantum Bayesian network, or quantum causal network where constraints can be imposed directly on quantum systems \cite{FritzII,Leifer1,Pienaar,Costa,Allen}. Note that causal structure is distinct from the approach of associating a causal network with certain processes. In particular, given a specific causal structure, there are multiple causal networks compatible with it, but a causal network will often fix allowed causal structures. 

Other related work includes the derivation by Markopoulou and others of covariance in discrete causal quantum structures \cite{Markopoulou, BIP}, or more general structures akin to ones considered in this paper \cite{CRCaucat}.

\section{Process theories}

By a \em  process theory \em we mean a collection of \em systems\em , represented by \em wires\em, and \em processes\em, represented by \em boxes \em with wires as inputs (at the bottom of the box) and outputs (at the top)  \cite{CKpaperI, CKBook}.  Moreover, when we plug these boxes together: 
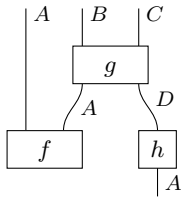
\[
\begin{tikzpicture}
	\begin{pgfonlayer}{nodelayer}
		\node [style={medium box}] (0) at (0, 1) {$g$};
		\node [style=none] (1) at (-1.25, -0.75) {};
		\node [style=none] (2) at (-0.75, 0.5) {};
		\node [style=none] (3) at (-0.75, 2.5) {};
		\node [style=none] (4) at (-0.75, 1.5) {};
		\node [style={right label}] (5) at (-2, 2.25) {$A$};  
		\node [style=none] (6) at (0.75, 0.5) {};
		\node [style=none] (7) at (1.25, -0.75) {};
		\node [style={medium box}] (8) at (-1.75, -1.25) {$f$};
		\node [style=none] (9) at (-2.25, -0.75) {};
		\node [style=none] (10) at (-2.25, 2.5) {};
		\node [style=none] (11) at (1.25, -2.5) {};
		\node [style=none] (12) at (1.25, -1.75) {};
		\node [style={small box}] (13) at (1.25, -1.25) {$h$};
		\node [style=none] (14) at (0.75, 1.5) {};
		\node [style=none] (15) at (0.75, 2.5) {};
		\node [style={right label}] (16) at (-0.5, 2.25) {$B$}; 
		\node [style={right label}] (17) at (1, 2.25) {$C$};
		\node [style={right label}] (18) at (-0.75, -0.25) {$A$};
		\node [style={right label}] (19) at (1.25, 0) {$D$};
		\node [style={right label}] (20) at (1.5, -2.25) {$A$};
	\end{pgfonlayer}
	\begin{pgfonlayer}{edgelayer}
		\draw [in=-90, out=90, looseness=1.00] (1.center) to (2.center);
		\draw (4.center) to (3.center);
		\draw [in=-90, out=90, looseness=1.00] (7.center) to (6.center);
		\draw (9.center) to (10.center);
		\draw (11.center) to (12.center);
		\draw (14.center) to (15.center);
	\end{pgfonlayer}
\end{tikzpicture}
\]
the resulting \em diagram \em should also be a process. 

For the purposes of this paper, outputs should be connected to inputs, and   diagrams should never contain causal loops. A diagram of processes with no directed cycles is called a \textit{circuit}.  A \em state \em is a process without inputs, and an \em effect \em is a process without outputs.  Processes without inputs or outputs are called \textit{numbers}. In the case of classical and quantum theory, these represent probabilities. 

Diagrams in process theories come with the interpretation that wires allow for a flow of information.  Hence, if two processes are disconnected: 
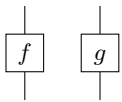
\[
\begin{tikzpicture}
	\begin{pgfonlayer}{nodelayer}
		\node [style=none] (0) at (-1, 1.25) {};
		\node [style=none] (1) at (-1, 0.5) {};
		\node [style=none] (2) at (-1, -1.25) {};
		\node [style=none] (3) at (-1, -0.5) {};
		\node [style=small box] (4) at (-1, 0) {$f$};
		\node [style=small box] (5) at (1, 0) {$g$};
		\node [style=none] (6) at (1, 1.25) {};
		\node [style=none] (7) at (1, -0.5) {};
		\node [style=none] (8) at (1, -1.25) {};
		\node [style=none] (9) at (1, 0.5) {};
	\end{pgfonlayer}
	\begin{pgfonlayer}{edgelayer}
		\draw [in=-90, out=90, looseness=1.00] (1.center) to (0.center);
		\draw (2.center) to (3.center);
		\draw [in=-90, out=90, looseness=1.00] (9.center) to (6.center);
		\draw (8.center) to (7.center);
	\end{pgfonlayer}
\end{tikzpicture}
\]
there cannot be any information exchange between them. This makes diagrams a useful mathematical tool to encode causation \cite{CRCaucat}, generalising notions such as Bayesian networks~\cite{PearlBook}, which can be seen as such circuits whose states are probability distributions and processes are conditional probability distributions.

Generalising states and processes to include density operators and CP-maps, one can capture the essential structure of quantum theory, including quantum-classical interaction and complementarity~\cite{CKBook}.

\section{Process terminality}

We will  assume that for each system in a process theory there exists a designated \em discarding \em effect:
\[
\begin{tikzpicture}
	\begin{pgfonlayer}{nodelayer}
		\node [style=none] (0) at (0, -0.75) {};
		\node [style=upground] (1) at (0, 0.75) {};
	\end{pgfonlayer}
	\begin{pgfonlayer}{edgelayer}
		\draw [style=swap] (1) to (0.center); 
	\end{pgfonlayer}
\end{tikzpicture}
\]
For a process, \em terminality \em means that (\ref{eq:term}) holds, and for a process theory, it means that it holds for all processes. 

When viewing probability theory as a process theory,  
discarding a classical system $A$ amounts to summing over all of its possible values~\cite{CKBook}. Hence, terminality for states means that $\sum_a P(A = a) = 1$, and terminality for processes means that for all $a \in A$, $\sum_b P(B = b|A = a) = 1$.

When viewing quantum theory as a process theory,   discarding  of quantum systems is  the trace \cite{CKBook}.  Hence, terminality means that density matrices have diagonal one, and that CP-maps are trace-preserving.  


%

\section{causal structure}

We will now explicitly encode the condition that only events in the causal past of a particular event can have an impact on that event. 

To any causal structure, we can associate a diagram of processes as follows. For each node $A$ in the causal structure, we introduce a new box (also labelled $A$) with an input $\inp A$ and an output $\outp A$. Then, we introduce a wire from an $A$-labelled box to a $B$-labelled box if and only if $B$ is a successor of $A$ in the causal structure:
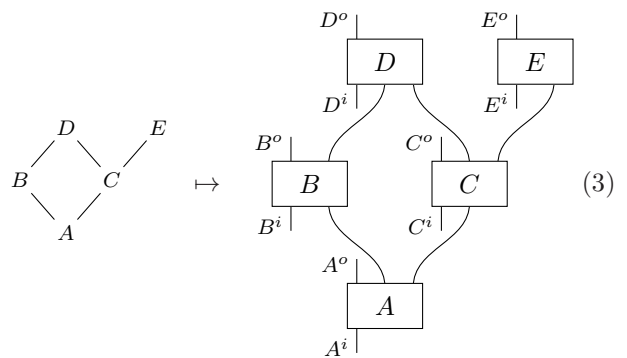
\begin{equation}\label{eq:caus-circ}
\begin{tikzpicture}[scale=0.7]
	\begin{pgfonlayer}{nodelayer}
		\node [style=label, minimum width=1 cm, inner sep=1 pt] (0) at (-1.75, 0) {$B$};
		\node [style=label, minimum width=1 cm, inner sep=1 pt] (1) at (0, 2) {$D$};
		\node [style=label, minimum width=1 cm, inner sep=1 pt] (2) at (1.75, 0) {$C$};
		\node [style=label, minimum width=1 cm, inner sep=1 pt] (3) at (3.5, 2) {$E$};
		\node [style=label, minimum width=1 cm, inner sep=1 pt] (4) at (0, -2) {$A$};
	\end{pgfonlayer}
	\begin{pgfonlayer}{edgelayer}
		\draw (4) to (2);
		\draw (2) to (1);
		\draw (0) to (1);
		\draw (4) to (0);
		\draw (2) to (3);
	\end{pgfonlayer}
\end{tikzpicture} \!\!\mapsto\ \ \begin{tikzpicture}
	\begin{pgfonlayer}{nodelayer}
		\node [style=box, minimum width=1 cm] (0) at (-2, 0) {$B$};
		\node [style=box, minimum width=1 cm] (1) at (0, 3.25) {$D$};
		\node [style=box, minimum width=1 cm] (2) at (2.25, 0) {$C$};
		\node [style=box, minimum width=1 cm] (3) at (4, 3.25) {$E$};
		\node [style=box, minimum width=1 cm] (4) at (0, -3.25) {$A$};
		\node [style=none] (5) at (-2.5, 0.5) {};
		\node [style=none] (6) at (-2.5, 1.25) {};
		\node [style=none] (7) at (-2.5, -0.5) {};
		\node [style=none] (8) at (-2.5, -1.25) {};
		\node [style=left label] (9) at (-2.75, -1.25) {$\inp B$};
		\node [style=left label] (10) at (-2.75, 1) {$\outp B$};
		\node [style=left label] (11) at (-1, -2.25) {$\outp A$};
		\node [style=left label] (12) at (-1, -4.5) {$\inp A$};
		\node [style=none] (13) at (-0.75, -3.75) {};
		\node [style=none] (14) at (-0.75, -2) {};
		\node [style=none] (15) at (-0.75, -4.5) {};
		\node [style=none] (16) at (-0.75, -2.75) {};
		\node [style=none] (17) at (-1.5, -0.5) {};
		\node [style=none] (18) at (0, -2.75) {};
		\node [style=none] (19) at (0.75, -2.75) {};
		\node [style=none] (20) at (2.25, -0.5) {};
		\node [style=left label] (21) at (1.25, 1) {$\outp C$};
		\node [style=left label] (22) at (1.25, -1.25) {$\inp C$};
		\node [style=none] (23) at (1.5, -0.5) {};
		\node [style=none] (24) at (1.5, 1.25) {};
		\node [style=none] (25) at (1.5, -1.25) {};
		\node [style=none] (26) at (1.5, 0.5) {};
		\node [style=none] (27) at (-0.75, 4.5) {};
		\node [style=left label] (28) at (-1, 4.25) {$\outp D$};
		\node [style=none] (29) at (-0.75, 2.75) {};
		\node [style=none] (30) at (-0.75, 2) {};
		\node [style=left label] (31) at (-1, 2) {$\inp D$};
		\node [style=none] (32) at (-0.75, 3.75) {};
		\node [style=none] (33) at (3.5, 4.5) {};
		\node [style=left label] (34) at (3.25, 4.25) {$\outp E$};
		\node [style=none] (35) at (3.5, 2.75) {};
		\node [style=none] (36) at (3.5, 2) {};
		\node [style=left label] (37) at (3.25, 2) {$\inp E$};
		\node [style=none] (38) at (3.5, 3.75) {};
		\node [style=none] (39) at (4.5, 2.75) {};
		\node [style=none] (40) at (3, 0.5) {};
		\node [style=none] (41) at (-1.5, 0.5) {};
		\node [style=none] (42) at (0, 2.75) {};
		\node [style=none] (43) at (2.25, 0.5) {};
		\node [style=none] (44) at (0.75, 2.75) {};
	\end{pgfonlayer}
	\begin{pgfonlayer}{edgelayer}
		\draw (5.center) to (6.center);
		\draw (8.center) to (7.center);
		\draw (16.center) to (14.center);
		\draw (15.center) to (13.center);
		\draw [in=-90, out=90, looseness=1.00] (18.center) to (17.center);
		\draw [in=-90, out=90, looseness=1.00] (19.center) to (20.center);
		\draw (26.center) to (24.center);
		\draw (25.center) to (23.center);
		\draw (32.center) to (27.center);
		\draw (30.center) to (29.center);
		\draw (38.center) to (33.center);
		\draw (36.center) to (35.center);
		\draw [in=90, out=-90, looseness=1.00] (39.center) to (40.center);
		\draw [in=-90, out=90, looseness=1.00] (41.center) to (42.center);
		\draw [in=-90, out=90, looseness=1.00] (43.center) to (44.center);
	\end{pgfonlayer}
\end{tikzpicture}\!\!
\end{equation}
We say an event $A$ is capable of \em affecting \em an event $B$ if and only if, by changing the input at $\inp A$ one is able to affect the output at $\outp B$. Operationally, we can interpret this an observer at $A$ being able to send a message to an observer at $B$. However, we can also choose to treat this as a purely formal notion of
a happening at $A$ impacting a happening at $B$, without making reference to an observer.

It should be the case the the only events capable of affecting a given event are those in its causal past.
However, by simply  considering diagrams of processes which decompose according to a given causal structure as in \eqref{eq:caus-circ} 
without any further constraints, it is easy to pick processes in some process theory where an event $B$ can be affected by events not in its causal past. For example, ignoring $A$'s input and $B$'s output, by means of a Bell state and a Bell effect, written as a cap-shaped and a cup-shaped wire respectively \cite{Kindergarten, CKpaperI, CKBook}, there is a clear information flow from future to past:  
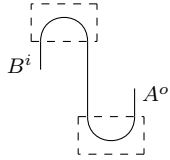
\ctikzfig{caus-tele}  
The underlying reason is of course that Bell-effects cannot be performed with certainty, and hence do not constitute a single box in a causal theory. 

We will constrain process theories such that the output of a given process only depends on what happened in the causal past of that process, and hence, for example, they will not contain Bell-effects. We will say that these process theories \em respect causal structure.\em



Suppose we have a circuit 
$\bm\Phi$ whose boxes are:  
\[
\mathcal B := \{A, B, C, \ldots \}
\]
with each box $A$ having input $\inp A$ and output $\outp A$. 
One can represent this circuit as one big box:
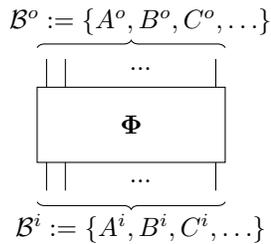
\ctikzfig{proc}

Now, write $A \leq B$ if $A = B$ or there exists a directed path of wires from box $A$ to box $B$ in the circuit $\bm\Phi$. Then for any subset of boxes $\mathcal E \subseteq \mathcal B$, let $\past(\mathcal E)$ be the boxes in the 
causal past of $\mathcal E$, i.e.
\[
\past(\mathcal E) := \{ A \,|\, \exists B \in \mathcal E : A \leq B \} 
\]
Then the outputs $\outp{\mathcal E}$ of the boxes in $\mathcal E$ should only be affected by the inputs $\inp{\past(\mathcal E)}$ of the boxes in $\past(\mathcal E)$.  In other words, discarding all of the outputs of $\bm\Phi$ except for $\outp{\mathcal E}$ should give the following factorisation:
\begin{equation}\label{eq:causality}
\begin{tikzpicture}
	\begin{pgfonlayer}{nodelayer}
		\node [style=none] (0) at (-5.25, -2.75) {$\inp{\past(\mathcal E)}$};
		\node [style=none] (1) at (-4.75, 2.75) {$\outp{\mathcal E}$};
		\node [style=box, minimum width=2.5 cm, minimum height=1 cm] (2) at (-3.5, 0) {$\bm\Phi$};
		\node [style=none] (3) at (-5.75, -0.75) {};
		\node [style=none] (4) at (-3.75, -0.75) {};
		\node [style=none] (5) at (-3.25, -0.75) {};
		\node [style=none] (6) at (-1.25, -0.75) {};
		\node [style=none] (7) at (-5.75, -1.75) {};
		\node [style=none] (8) at (-3.25, -1.75) {};
		\node [style=none] (9) at (-1.25, -1.75) {};
		\node [style=none] (10) at (-3.75, -1.75) {};
		\node [style=none] (11) at (-4.75, -1.5) {...};
		\node [style=none] (12) at (-2.25, -1.25) {...};
		\node [style=none] (13) at (-3.25, 0.5) {};
		\node [style=none] (14) at (-4.75, 1.5) {...};
		\node [style=none] (15) at (-1.25, 0.5) {};
		\node [style=upground] (16) at (-1.25, 2.75) {};
		\node [style=none] (17) at (-3.75, 0.5) {};
		\node [style=none] (18) at (-2.25, 1.75) {...};
		\node [style=none] (19) at (-5.75, 0.5) {};
		\node [style=none] (20) at (-3.75, 1.75) {};
		\node [style=upground] (21) at (-3.25, 2.75) {};
		\node [style=none] (22) at (-5.75, 1.75) {};
		\node [style=none] (23) at (-3.5, 2.25) {};
		\node [style=none] (24) at (-6, 2.25) {};
		\node [style=none] (25) at (-6, -2.25) {};
		\node [style=none] (26) at (-3.5, -2.25) {};
		\node [style=none] (27) at (0, 0) {$=$};
		\node [style=none] (28) at (4.5, -1.75) {};
		\node [style=none] (29) at (1.25, -1.75) {};
		\node [style=none] (30) at (3.5, 2.25) {};
		\node [style=none] (31) at (2.25, -1.5) {...};
		\node [style=none] (32) at (3.25, -0.75) {};
		\node [style=none] (33) at (1.25, 1.75) {};
		\node [style=none] (34) at (1.75, -2.75) {$\inp{\past(\mathcal E)}$};
		\node [style=upground] (35) at (6.5, 0) {};
		\node [style=none] (36) at (3.25, -1.75) {};
		\node [style=none] (37) at (6.5, -1.75) {};
		\node [style=none] (38) at (2.25, 2.75) {$\outp{\mathcal E}$};
		\node [style=none] (39) at (3.5, -2.25) {};
		\node [style=none] (40) at (1.25, 0.5) {};
		\node [style=none] (41) at (5.5, -1) {...};
		\node [style=none] (42) at (1, 2.25) {};
		\node [style=upground] (43) at (4.5, 0) {};
		\node [style=none] (44) at (1, -2.25) {};
		\node [style=none] (45) at (1.25, -0.75) {};
		\node [style=none] (46) at (2.25, 1.5) {...};
		\node [style=box, minimum width=1.25 cm, minimum height=1 cm] (47) at (2.25, 0) {$\bm\Phi'$};
		\node [style=none] (48) at (3.25, 1.75) {};
		\node [style=none] (49) at (3.25, 0.5) {};
	\end{pgfonlayer}
	\begin{pgfonlayer}{edgelayer}
		\draw (7.center) to (3.center);
		\draw (10.center) to (4.center);
		\draw (8.center) to (5.center);
		\draw (9.center) to (6.center);
		\draw (22.center) to (19.center);
		\draw (20.center) to (17.center);
		\draw (21) to (13.center);
		\draw (16) to (15.center);
		\draw [style=small braceedge] (24.center) to (23.center);
		\draw [style=small braceedge] (26.center) to (25.center);
		\draw (29.center) to (45.center);
		\draw (36.center) to (32.center);
		\draw (33.center) to (40.center);
		\draw (48.center) to (49.center);
		\draw (43) to (28.center);
		\draw (35) to (37.center);
		\draw [style=small braceedge] (42.center) to (30.center);
		\draw [style=small braceedge] (39.center) to (44.center);
	\end{pgfonlayer}
\end{tikzpicture}
\end{equation}
for some process process $\bm\Phi'$. Indeed, concerning outputs $\outp{\mathcal E}$ (cf.~we discard the other outputs in LHS), the inputs not in $\inp{\past(\mathcal E)}$ must be discarded (cf.~RHS) as they may not affect $\outp{\mathcal E}$. Hence we can conclude that a process theory respects causal structure if \eqref{eq:causality} is satisfied for all circuits $\bm\Phi$ and subsets $\mathcal E$.

\section{Proof of equivalence}

\begin{theorem*}
  A process theory respects causal structure if and only if it is terminal.
\end{theorem*}

\begin{proof}
  First suppose a process theory respects causal structure. Then, we can treat every process as a circuit $\bm\Phi$ for a single node $A$:
  \ctikzfig{single-proc_copy}
  The set of boxes of $\bm\Phi$ is the singleton $\{ A \}$. The only subset of $\mathcal B$ giving a non-trivial expression for \eqref{eq:causality} is $\mathcal E := \emptyset$. This gives:
  \ctikzfig{single-proc-caus_copy}
  which is exactly the terminality equation \eqref{eq:term}.

  Conversely, suppose every process in a process theory satisfies the terminality equation \eqref{eq:term}.  Fix a circuit $\bm\Phi$, 
  a subset $\mathcal E \subseteq \mathcal B$, and let ${\bm\Phi}|_{\mathcal E}$ be the sub-circuit of $\bm\Phi$ whose boxes are in $\past(\mathcal E)$. For example, taking $\bm\Phi$ from \eqref{eq:caus-circ}, and $\mathcal E := \{ B, C \}$ yields:
  
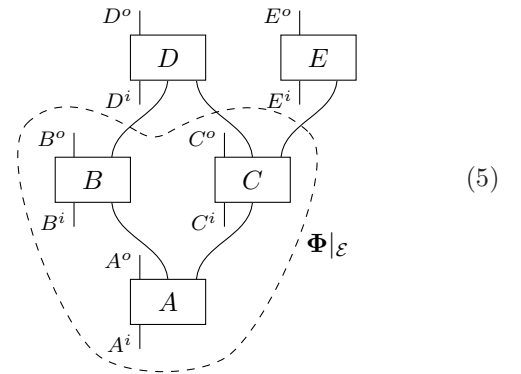
\begin{equation}\label{eq:sub-circ}
    \begin{tikzpicture}
	\begin{pgfonlayer}{nodelayer}
		\node [style=box, minimum width=1 cm] (0) at (-2, 0) {$B$};
		\node [style=box, minimum width=1 cm] (1) at (0, 3.25) {$D$};
		\node [style=box, minimum width=1 cm] (2) at (2.25, 0) {$C$};
		\node [style=box, minimum width=1 cm] (3) at (4, 3.25) {$E$};
		\node [style=box, minimum width=1 cm] (4) at (0, -3.25) {$A$};
		\node [style=none] (5) at (-2.5, 0.5) {};
		\node [style=none] (6) at (-2.5, 1.25) {};
		\node [style=none] (7) at (-2.5, -0.5) {};
		\node [style=none] (8) at (-2.5, -1.25) {};
		\node [style=left label] (9) at (-2.75, -1.25) {$\inp B$};
		\node [style=left label] (10) at (-2.75, 1) {$\outp B$};
		\node [style=left label] (11) at (-1, -2.25) {$\outp A$};
		\node [style=left label] (12) at (-1, -4.5) {$\inp A$};
		\node [style=none] (13) at (-0.75, -3.75) {};
		\node [style=none] (14) at (-0.75, -2) {};
		\node [style=none] (15) at (-0.75, -4.5) {};
		\node [style=none] (16) at (-0.75, -2.75) {};
		\node [style=none] (17) at (-1.5, -0.5) {};
		\node [style=none] (18) at (0, -2.75) {};
		\node [style=none] (19) at (0.75, -2.75) {};
		\node [style=none] (20) at (2.25, -0.5) {};
		\node [style=left label] (21) at (1.25, 1) {$\outp C$};
		\node [style=left label] (22) at (1.25, -1.25) {$\inp C$};
		\node [style=none] (23) at (1.5, -0.5) {};
		\node [style=none] (24) at (1.5, 1.25) {};
		\node [style=none] (25) at (1.5, -1.25) {};
		\node [style=none] (26) at (1.5, 0.5) {};
		\node [style=none] (27) at (-0.75, 4.5) {};
		\node [style=left label] (28) at (-1, 4.25) {$\outp D$};
		\node [style=none] (29) at (-0.75, 2.75) {};
		\node [style=none] (30) at (-0.75, 2) {};
		\node [style=left label] (31) at (-1, 2) {$\inp D$};
		\node [style=none] (32) at (-0.75, 3.75) {};
		\node [style=none] (33) at (3.5, 4.5) {};
		\node [style=left label] (34) at (3.25, 4.25) {$\outp E$};
		\node [style=none] (35) at (3.5, 2.75) {};
		\node [style=none] (36) at (3.5, 2) {};
		\node [style=left label] (37) at (3.25, 2) {$\inp E$};
		\node [style=none] (38) at (3.5, 3.75) {};
		\node [style=none] (39) at (4.5, 2.75) {};
		\node [style=none] (40) at (3, 0.5) {};
		\node [style=none] (41) at (-1.5, 0.5) {};
		\node [style=none] (42) at (0, 2.75) {};
		\node [style=none] (43) at (2.25, 0.5) {};
		\node [style=none] (44) at (0.75, 2.75) {};
		\node [style=none] (45) at (-3.75, 1.5) {};
		\node [style=none] (46) at (-0.75, 1.25) {};
		\node [style=none] (47) at (0, 1.25) {};
		\node [style=none] (48) at (3.75, 1.25) {};
		\node [style=none] (49) at (1.5, -4.75) {};
		\node [style=none] (50) at (-1.75, -4.75) {};
		\node [style=none] (51) at (4.25, -1.75) {${\bm\Phi}|_{\mathcal E}$};
	\end{pgfonlayer}
	\begin{pgfonlayer}{edgelayer}
		\draw (5.center) to (6.center);
		\draw (8.center) to (7.center);
		\draw (16.center) to (14.center);
		\draw (15.center) to (13.center);
		\draw [in=-90, out=90, looseness=1.00] (18.center) to (17.center);
		\draw [in=-90, out=90, looseness=1.00] (19.center) to (20.center);
		\draw (26.center) to (24.center);
		\draw (25.center) to (23.center);
		\draw (32.center) to (27.center);
		\draw (30.center) to (29.center);
		\draw (38.center) to (33.center);
		\draw (36.center) to (35.center);
		\draw [in=90, out=-90, looseness=1.00] (39.center) to (40.center);
		\draw [in=-90, out=90, looseness=1.00] (41.center) to (42.center);
		\draw [in=-90, out=90, looseness=1.00] (43.center) to (44.center);
		\draw [style=dashed edge, in=-30, out=-150, looseness=0.75] (49.center) to (50.center);
		\draw [style=dashed edge, in=-135, out=150, looseness=0.50] (50.center) to (45.center);
		\draw [style=dashed edge, in=150, out=45, looseness=1.00] (45.center) to (46.center);
		\draw [style=dashed edge, in=-150, out=-30, looseness=1.00] (46.center) to (47.center);
		\draw [style=dashed edge, in=135, out=30, looseness=1.00] (47.center) to (48.center);
		\draw [style=dashed edge, in=30, out=-45, looseness=0.50] (48.center) to (49.center);
	\end{pgfonlayer}
\end{tikzpicture}
  \end{equation}
  Then, let $\bm\Phi'$ be ${\bm\Phi}|_{\mathcal E}$ with all of the outputs not in $\outp{\mathcal E}$ discarded. In the example above, $\outp{\mathcal E} = \{ \outp B, \outp C \}$, so:
  \ctikzfig{sub-circ-F}
  Since $\mathcal E \subseteq \past(\mathcal E)$, the outputs of $\bm\Phi'$ will be $\outp{\mathcal E}$, and since $\past(\mathcal E)$ is  \textit{downward closed} (i.e. it contains all of the boxes in its causal past), its inputs will be $\inp{\past(\mathcal E)}$.
   Thus it remains to show that discarding the outputs $\outp{\bm\Phi} \backslash \outp{\mathcal E}$ causes the boxes in $\mathcal B \backslash \past(\mathcal E)$ to vanish, leaving their inputs discarded.

   Call a box \textit{maximal} in the circuit if its outputs are not connected to another box. Since $\past(\mathcal E)$ is downward closed, then either $\mathcal B \backslash \past(\mathcal E)$ contains a maximal box or it is empty.
   If $\mathcal B \backslash \past(\mathcal E)$ is empty we are done, so suppose it contains a maximal box $A$.
   Since none of the outputs of $A$ are in $\outp{\mathcal E}$, they are all discarded in the LHS of \eqref{eq:causality}. Hence we can apply \eqref{eq:term} to remove $A$ and discard its inputs.
   For example, we can choose a maximal box $E \in \mathcal B \backslash \past(\mathcal E)$ in diagram \eqref{eq:sub-circ} and apply Equation~\eqref{eq:term} to obtain:
   
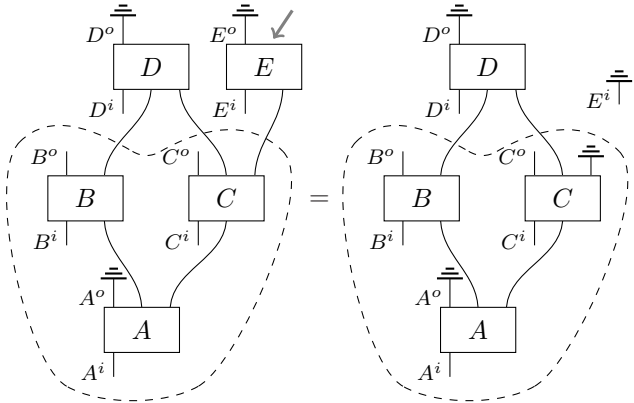
\[ \begin{tikzpicture}
	\path [use as bounding box] (-4,-5.5) rectangle (4,5.25);
	\begin{pgfonlayer}{nodelayer}
		\node [style=box, minimum width=1 cm] (0) at (-1.75, 0) {$B$};
		\node [style=box, minimum width=1 cm] (1) at (0, 3.5) {$D$};
		\node [style=box, minimum width=1 cm] (2) at (2, 0) {$C$};
		\node [style=box, minimum width=1 cm] (3) at (3, 3.5) {$E$};
		\node [style=box, minimum width=1 cm] (4) at (-0.25, -3.5) {$A$};
		\node [style=none] (5) at (-2.25, 0.5) {};
		\node [style=none] (6) at (-2.25, 1.25) {};
		\node [style=none] (7) at (-2.25, -0.5) {};
		\node [style=none] (8) at (-2.25, -1.25) {};
		\node [style=left label] (9) at (-2.5, -1.25) {$\inp B$};
		\node [style=left label] (10) at (-2.5, 1) {$\outp B$};
		\node [style=left label] (11) at (-1.25, -2.75) {$\outp A$};
		\node [style=left label] (12) at (-1.25, -4.75) {$\inp A$};
		\node [style=none] (13) at (-1, -4) {};
		\node [style=upground] (14) at (-1, -2) {};
		\node [style=none] (15) at (-1, -4.75) {};
		\node [style=none] (16) at (-1, -3) {};
		\node [style=none] (17) at (-1.25, -0.5) {};
		\node [style=none] (18) at (-0.25, -3) {};
		\node [style=none] (19) at (0.5, -3) {};
		\node [style=none] (20) at (2, -0.5) {};
		\node [style=left label] (21) at (1, 1) {$\outp C$};
		\node [style=left label] (22) at (1, -1.25) {$\inp C$};
		\node [style=none] (23) at (1.25, -0.5) {};
		\node [style=none] (24) at (1.25, 1.25) {};
		\node [style=none] (25) at (1.25, -1.25) {};
		\node [style=none] (26) at (1.25, 0.5) {};
		\node [style=upground] (27) at (-0.75, 5) {};
		\node [style=left label] (28) at (-1, 4.25) {$\outp D$};
		\node [style=none] (29) at (-0.75, 3) {};
		\node [style=none] (30) at (-0.75, 2.25) {};
		\node [style=left label] (31) at (-1, 2.25) {$\inp D$};
		\node [style=none] (32) at (-0.75, 4) {};
		\node [style=upground] (33) at (2.5, 5) {};
		\node [style=left label] (34) at (2.25, 4.25) {$\outp E$};
		\node [style=none] (35) at (2.5, 3) {};
		\node [style=none] (36) at (2.5, 2.25) {};
		\node [style=left label] (37) at (2.25, 2.25) {$\inp E$};
		\node [style=none] (38) at (2.5, 4) {};
		\node [style=none] (39) at (3.5, 3) {};
		\node [style=none] (40) at (2.75, 0.5) {};
		\node [style=none] (41) at (-1.25, 0.5) {};
		\node [style=none] (42) at (0, 3) {};
		\node [style=none] (43) at (2, 0.5) {};
		\node [style=none] (44) at (0.75, 3) {};
		\node [style=none] (45) at (-3.5, 1.5) {};
		\node [style=none] (46) at (-0.5, 1.25) {};
		\node [style=none] (47) at (0.25, 1.25) {};
		\node [style=none] (48) at (3.5, 1.25) {};
		\node [style=none] (49) at (1.25, -5) {};
		\node [style=none] (50) at (-2.25, -5) {};
		\node [style=none] (51) at (3.25, 4.25) {};
		\node [style=none] (52) at (3.75, 5) {};
	\end{pgfonlayer}
	\begin{pgfonlayer}{edgelayer}
		\draw (5.center) to (6.center);
		\draw (8.center) to (7.center);
		\draw (16.center) to (14);
		\draw (15.center) to (13.center);
		\draw [in=-90, out=90, looseness=1.00] (18.center) to (17.center);
		\draw [in=-90, out=90, looseness=1.00] (19.center) to (20.center);
		\draw (26.center) to (24.center);
		\draw (25.center) to (23.center);
		\draw (32.center) to (27);
		\draw (30.center) to (29.center);
		\draw (38.center) to (33);
		\draw (36.center) to (35.center);
		\draw [in=90, out=-90, looseness=1.00] (39.center) to (40.center);
		\draw [in=-90, out=90, looseness=1.00] (41.center) to (42.center);
		\draw [in=-90, out=90, looseness=1.00] (43.center) to (44.center);
		\draw [style=dashed edge, in=-30, out=-150, looseness=0.75] (49.center) to (50.center);
		\draw [style=dashed edge, in=-135, out=150, looseness=0.50] (50.center) to (45.center);
		\draw [style=dashed edge, in=150, out=45, looseness=1.00] (45.center) to (46.center);
		\draw [style=dashed edge, in=-150, out=-30, looseness=1.00] (46.center) to (47.center);
		\draw [style=dashed edge, in=135, out=30, looseness=1.00] (47.center) to (48.center);
		\draw [style=dashed edge, in=30, out=-45, looseness=0.50] (48.center) to (49.center);
		\draw [style=pointer edge] (52.center) to (51.center);
	\end{pgfonlayer}
\end{tikzpicture}=\begin{tikzpicture}
	\path [use as bounding box] (-4,-5.5) rectangle (4,5.25);
	\begin{pgfonlayer}{nodelayer}
		\node [style=box, minimum width=1 cm] (0) at (-1.75, 0) {$B$};
		\node [style=box, minimum width=1 cm] (1) at (0, 3.5) {$D$};
		\node [style=box, minimum width=1 cm] (2) at (2, 0) {$C$};
		\node [style=box, minimum width=1 cm] (3) at (-0.25, -3.5) {$A$};
		\node [style=none] (4) at (-2.25, 0.5) {};
		\node [style=none] (5) at (-2.25, 1.25) {};
		\node [style=none] (6) at (-2.25, -0.5) {};
		\node [style=none] (7) at (-2.25, -1.25) {};
		\node [style=left label] (8) at (-2.5, -1.25) {$\inp B$};
		\node [style=left label] (9) at (-2.5, 1) {$\outp B$};
		\node [style=left label] (10) at (-1.25, -2.75) {$\outp A$};
		\node [style=left label] (11) at (-1.25, -4.75) {$\inp A$};
		\node [style=none] (12) at (-1, -4) {};
		\node [style=upground] (13) at (-1, -2) {};
		\node [style=none] (14) at (-1, -4.75) {};
		\node [style=none] (15) at (-1, -3) {};
		\node [style=none] (16) at (-1.25, -0.5) {};
		\node [style=none] (17) at (-0.25, -3) {};
		\node [style=none] (18) at (0.5, -3) {};
		\node [style=none] (19) at (2, -0.5) {};
		\node [style=left label] (20) at (1, 1) {$\outp C$};
		\node [style=left label] (21) at (1, -1.25) {$\inp C$};
		\node [style=none] (22) at (1.25, -0.5) {};
		\node [style=none] (23) at (1.25, 1.25) {};
		\node [style=none] (24) at (1.25, -1.25) {};
		\node [style=none] (25) at (1.25, 0.5) {};
		\node [style=upground] (26) at (-0.75, 5) {};
		\node [style=left label] (27) at (-1, 4.25) {$\outp D$};
		\node [style=none] (28) at (-0.75, 3) {};
		\node [style=none] (29) at (-0.75, 2.25) {};
		\node [style=left label] (30) at (-1, 2.25) {$\inp D$};
		\node [style=none] (31) at (-0.75, 4) {};
		\node [style=none] (32) at (-1.25, 0.5) {};
		\node [style=none] (33) at (0, 3) {};
		\node [style=none] (34) at (2, 0.5) {};
		\node [style=none] (35) at (0.75, 3) {};
		\node [style=none] (36) at (-3.5, 1.5) {};
		\node [style=none] (37) at (-0.5, 1.25) {};
		\node [style=none] (38) at (0.25, 1.25) {};
		\node [style=none] (39) at (3.5, 1.25) {};
		\node [style=none] (40) at (1.25, -5) {};
		\node [style=none] (41) at (-2.25, -5) {};
		\node [style=none] (42) at (3.5, 2.5) {};
		\node [style=upground] (43) at (3.5, 3.25) {};
		\node [style=upground] (44) at (2.75, 1.25) {};
		\node [style=none] (45) at (2.75, 0.5) {};
		\node [style=left label] (46) at (3.25, 2.5) {$\inp E$};
	\end{pgfonlayer}
	\begin{pgfonlayer}{edgelayer}
		\draw (4.center) to (5.center);
		\draw (7.center) to (6.center);
		\draw (15.center) to (13);
		\draw (14.center) to (12.center);
		\draw [in=-90, out=90, looseness=1.00] (17.center) to (16.center);
		\draw [in=-90, out=90, looseness=1.00] (18.center) to (19.center);
		\draw (25.center) to (23.center);
		\draw (24.center) to (22.center);
		\draw (31.center) to (26);
		\draw (29.center) to (28.center);
		\draw [in=-90, out=90, looseness=1.00] (32.center) to (33.center);
		\draw [in=-90, out=90, looseness=1.00] (34.center) to (35.center);
		\draw [style=dashed edge, in=-30, out=-150, looseness=0.75] (40.center) to (41.center);
		\draw [style=dashed edge, in=-135, out=150, looseness=0.50] (41.center) to (36.center);
		\draw [style=dashed edge, in=150, out=45, looseness=1.00] (36.center) to (37.center);
		\draw [style=dashed edge, in=-150, out=-30, looseness=1.00] (37.center) to (38.center);
		\draw [style=dashed edge, in=135, out=30, looseness=1.00] (38.center) to (39.center);
		\draw [style=dashed edge, in=30, out=-45, looseness=0.50] (39.center) to (40.center);
		\draw (42.center) to (43);
		\draw (45.center) to (44);
	\end{pgfonlayer}
\end{tikzpicture} \]
   The set $\past(\mathcal E)$ remains downward closed, so we can repeat this procedure, choosing the next maximal box not in $\past(\mathcal E)$ until nothing but $\bm\Phi'$ and the discarded inputs of $\mathcal B \backslash \past(\mathcal E)$ are left. Hence we obtain \eqref{eq:causality} for all circuits $\bm\Phi$, so the process theory respects causal structure.
\end{proof}

\section{Re-casting previous works}


Diagrams with associated causal structure:
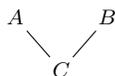
\ctikzfig{caus-ord-V}
where we moreover take the local input and output at $C$ to be trivial (i.e.~no input/output) result in the Bell-scenario typical for evaluating whether a theory is signalling:   
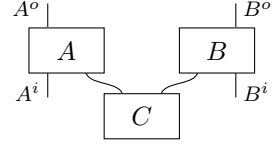
\ctikzfig{Bell}
When restricting  `respecting causal structure' in the statement of our theorem to these Bell-scenarios, the equivalence still holds, provided all diagrams with no inputs or outputs coincide. In this manner we recover the proof of no-signalling from terminality of \cite{Cnonsig,coecke2016terminality}.  

Similarly, we can recover generalised covariance as in \cite{CRCaucat,coecke2016terminality}, which generalised earlier results in \cite{Markopoulou, BIP}.

%
%

\section{Conclusion}

In (\ref{eq:term}) we have found a substitute for the idea of a causal structure, which has a very simple form, as well as a very clear interpretation,   and we pointed out that (\ref{eq:term})  
 is absolutely crucial to being able to do science at all.

Condition (\ref{eq:term}) moreover applies to arbitrary process theories and hence provides a starting point for crafting new causal theories. The upshot is that one does not need any prior commitment to a particular causal structure and can be readily generalised to conditions on processes which exhibit indefinite causal structure, as shown in~\cite{KU}. Such  processes include the quantum switch~\cite{QSwitch} and the causally non-separable processes of Oreskov, Costa, and Bruker~\cite{ViennaIndef}. Hence, this paves the way for crafting toy theories where there is an important interaction between causal structure and quantum structure, and ultimately, this could pave a way to a theory of quantum gravity.

\section{Acknowledgements}

This work is supported by the ERC under the European Union’s Seventh Framework Programme (FP7/2007-2013) / ERC grant n$^\text{o}$ 320571 and
a grant from the John Templeton Foundation. The opinions expressed in this publication are those of the author(s) and do not necessarily reflect the views of the John Templeton Foundation.

\bibliographystyle{apsrev}
\bibliography{main}
\end{document}